\documentclass[conference]{IEEEtran}
\usepackage{mathtools}
\usepackage{amsmath} 
\usepackage{amssymb} 
\usepackage{graphics} 
\usepackage{graphicx} 
\usepackage{amsthm}
\usepackage{cite}
\usepackage{epsfig}
\usepackage{epstopdf}

\makeatletter
\renewcommand*\env@matrix[1][*\c@MaxMatrixCols c]{%
	\hskip -\arraycolsep
	\let\@ifnextchar\new@ifnextchar
	\array{#1}}
\makeatother

\newcommand{\Aa}{\mathcal{A}}
\newcommand{\Bb}{\mathcal{B}}
\newcommand{\aA}{\mathbf{A}}
\newcommand{\bB}{\mathbf{B}}
\newcommand{\bbB}{\bar{\mathbf{B}}}
\newcommand{\yy}{\mathbf{y}}
\newcommand{\xx}{\mathbf{x}}
\newcommand{\vv}{\mathbf{v}}
\newcommand{\ff}{\mathbf{f}}
\newcommand{\aaa}{\mathbf{a}}
\newcommand{\bbb}{\mathbf{b}}
\newcommand{\HH}{\mathbf{H}}
\newcommand{\hHH}{\bar{\mathbf{H}}}

\newcommand{\FF}{\mathbf{F}}
\newcommand{\GG}{\mathbf{G}}
\newcommand{\gGG}{\bar{\mathbf{G}}}
\newcommand{\VV}{\mathbf{V}}
\newcommand{\ww}{\mathbf{w}}

\newcommand{\UU}{\mathbf{U}}
\newcommand{\II}{\mathbf{I}}

\newcommand{\CN}{\mathcal{CN}}

\newcommand{\C}{\mathbb{C}}
\newcommand{\Grass}{G_{N,M}(\C)}

\newcommand{\Grassss}{G_{2M,M}(\C)}
\newcommand{\OO}{\mathbf{O}}
\theoremstyle{definition}
\newtheorem{Lem}{Lemma}

\begin{document}
\title{Low Complexity Opportunistic Interference Alignment in $K$-Transmitter MIMO Interference Channels}


\author{\IEEEauthorblockN{Atul Kumar Sinha}
\IEEEauthorblockA{Department of Electrical Engineering\\
Indian Institute of Technology\\
Kanpur, India - 208016\\
Email: atulkumarin@gmail.com}
\and
\IEEEauthorblockN{A.~K.~Chaturvedi}
\IEEEauthorblockA{Department of Electrical Engineering\\
	Indian Institute of Technology\\
	Kanpur, India - 208016\\
	Email: akc@iitk.ac.in}}

\maketitle

\begin{abstract}
In this paper, we propose low complexity opportunistic methods for interference alignment in $K$-transmitter MIMO interference channels by exploiting multiuser diversity. We do not assume availability of channel state information (CSI) at the transmitters. Receivers are required to feed back analog values indicating the extent to which the received interference subspaces are aligned. The proposed opportunistic interference alignment (OIA) achieves sum-rate comparable to conventional OIA schemes but with a significantly reduced computational complexity. 
\end{abstract}


\IEEEpeerreviewmaketitle

\section{Introduction}
Interference alignment (IA) is a promising interference management technique for future wireless networks which are interference limited, such as, MIMO interference channels (IC), MIMO interfering broadcast channels (IFBC), etc. It was demonstrated in \cite{cadambe2008} that IA can achieve a sum degree-of-freedom (DoF) of $\frac{K}{2}$ in a $K$-user SISO interference channel. IA utilizes multiple signaling dimensions (due to multiple antennas or time/frequency extensions) to suppress the received interference into a reduced dimensional subspace of the receive space.

Conventional methods for interference alignment \cite{cadambe2008, gomadam2011, peters2009, papailiopoulos2012} depend on one or more of global channel state information, channel state information at transmitters, reciprocity of downlink-uplink channels, transmitter cooperation or iterative methods. If these assumptions are relaxed, it is not possible to achieve interference alignment by employing transmit precoding in order to align the interferences received at the receivers and thus we rely on opportunistic methods to select users for which the interferences are naturally aligned.

Low overhead feed back based OIA has been proposed in \cite{OIA5, OIA1, OIA2}. A $2 \times 2$ MIMO IC with $3$-transmitters has been considered in \cite{OIA5}, while \cite{OIA1, OIA2} extend it to the case of $M \times 2M$ MIMO IC, again with $3$-transmitters. In these works, it is assumed each transmitter has a separate user group. In each group, a single user is selected by the corresponding transmitter for opportunistic IA. In this paper, we extend OIA for the general case of $K$-transmitter MIMO IC. Further, in order to better exploit the available user diversity, we consider the problem of user pairing for achieving opportunistic IA. We use a geometric interpretation of the signal space to define the {measure of alignment} which quantifies the extent to which the interference subspaces are aligned. Each transmitter broadcasts a reference signal and receivers calculate their corresponding measure of alignment and feed it back. Depending on the received values for the measure of alignment, the transmitters can select their user independently (user selection) or they can be paired with the users by a central node (user pairing). 

The remainder of this paper is organized as follows. In Section \ref{systemmodel}, we describe the system model. Section \ref{moa} describes the proposed choice for measure of alignment. Section \ref{oiaus} discusses about OIA in the user selection framework while Section \ref{oiaup} discusses OIA in the user pairing framework. Performance comparison is presented in Section \ref{results}. Finally, Section \ref{conclusion} concludes the paper.

\section{System Model}
\label{systemmodel}
We consider a network with $K$ transmitters and $N$ receivers. Each transmitter is equipped with $N_T$ antennas and each receiver is equipped with $N_R$ antennas. In a $K$-transmitter MIMO IC, each user receives $(K-1)$ interfering signals and one desired signal, each of dimension $N_T$. We let $N_T = M$ and $N_R = 2M$ so that $M$ dimensions can be designated for the desired data streams and the remaining $M$ dimensions for interference alignment, at each user. We consider two different system models, namely, user selection model and user pairing model.

\subsection{User Selection}
\label{systemmodelus}
In the user selection framework, we assume that there are $K$ cells, each with a single transmitter (base station). The receivers (users) are arbitrarily divided into $K$ groups of size ${S} = N/K$, where $N$ is the total number of users in the network. The signal $\yy_n^k \in \C^{N_R \times 1}$ received by the $n$th user in the $k$th cell is given by :
\begin{equation}
\label{eq:model1}
\yy_{n}^k = \underbrace{\HH_{n,k}^k\xx_k}_\text{desired signal}
+ \underbrace{\sum_{l \neq k}{\HH_{n,l}^k\xx_l}}_\text{interference signals} + \quad \ww_{n}^k
\end{equation}

where $\HH_{n,l}^k \in \C^{N_R \times N_T}$ is the channel gain matrix between $l$th base station (BS) and user $n$ in cell $k$ and with each entry assumed to be independently and identically distributed (i.i.d.) circular symmetric complex Gaussian (CSCG) random variable with unit variance $\CN(0,1)$, $\ww_n^k \in \C^{N_R \times 1}$ denotes an additive white Gaussian noise (AWGN) with $\ww_n^k \sim \CN(0, \II_{N_R})$ and $\xx_l \in \C^{N_T \times 1}$ is the signal vector transmitted by transmitter $l$, encoded by a Gaussian codebook.

\begin{figure}[t!]
	\centering
	\includegraphics[scale=0.36]{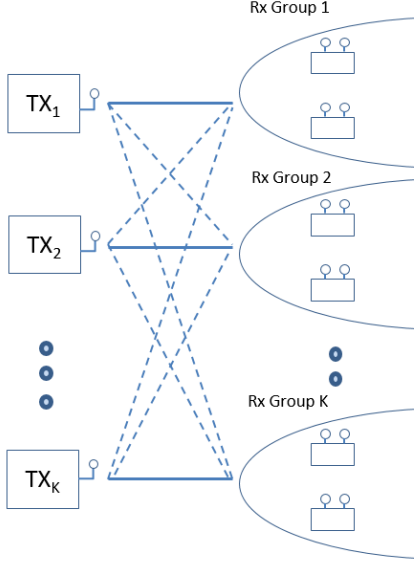} 
	\caption{User Selection Model : Each transmitter selects and serves a single user in each group}
\end{figure}

Fig.~1 depicts the user selection model. In each cell, out of a total of ${S}$ users, only one user is selected by the corresponding BS for data transmission. This selection is carried out such that opportunistic IA can be achieved and the exact procedure will be discussed in Section \ref{oiaus}.

\subsection{User Pairing}
\label{systemmodelup}

In the user pairing framework, we assume that each of the $K$ BSs can connect to any user in the network. In other words, BSs and users are not divided into distinct cells. For differentiating from the user selection model, we let $\GG_{n,k} \in \C^{N_R \times N_T}$ to be the channel gain matrix between BS $k$ and user $n$. The received signal $\yy_n \in \C^{N_R \times 1}$ at receiver $n$ is given by
\begin{equation}
\label{eq:modelchap5}
\yy_{n} = \sum_{k=1}^{K}{\GG_{n,k}\xx_k } + \ww_{n}
\end{equation}

with each entry of $\GG_{n,k}$ being i.i.d. CSCG random variable $\CN(0,1)$ and $\ww_n \in \C^{N_R \times 1}$ is AWGN at user $n$ with $\ww_n \sim \CN(0, \II_{N_R})$

We define the $N \times K$ pairing matrix $\mathbf{P}$ with $P_{i,j}$ as the entry on its $i$th row and $j$th column which is given by : 
\[
P_{i,j} = \begin{dcases*}
1, & if receiver $i$ and transmitter $j$ form a pair\\
0, & otherwise
\end{dcases*}
\]

where $\sum_{i}P_{i,j} = 1$ and $\sum_{j}P_{i,j} \leq 1$ to ensure that each BS is connected to exactly one user and that each `connected' user is connected to exactly one BS. The received signal at the $n$th user, $\yy_n$ can now be decomposed as :
\begin{align}
\label{eq:modelchap5pair}
\yy_{n} = \underbrace{\sum_{k=1}^{K}P_{n,k}\GG_{n,k}\xx_k}_\text{desired signal} 
+ \underbrace{\sum_{k=1}^{K}{(1- P_{n,k})\GG_{n,k}\xx_k}}_\text{interference signals} + \ww_{n}
\end{align}

\begin{figure}[t!]
	\centering
	\includegraphics[scale=0.475]{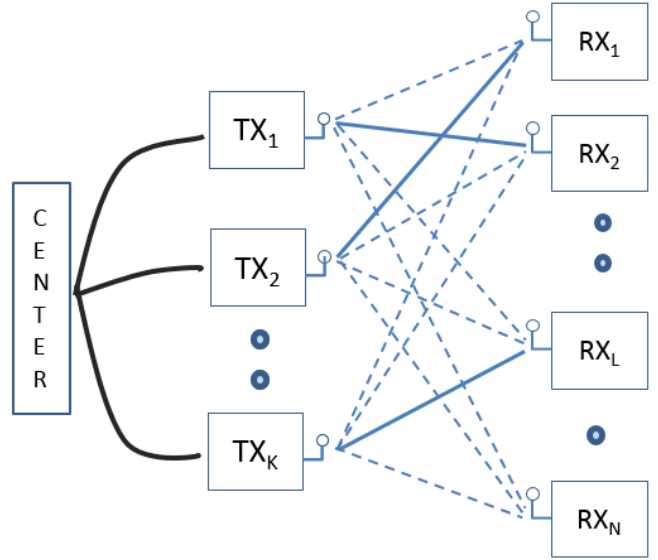} 
	\caption{User Pairing : Each transmitter is paired with one user by the center}
\end{figure}

Fig.~2 depicts the user pairing model. The pairing is enabled by the presence of a central system connecting each BS through backhaul links. The exact procedure for finding the transmitter-receiver pairing configuration or equivalently $\mathbf{P}$ such that opportunistic IA can be achieved, will be discussed in Section \ref{oiaup}. After the pairing configuration is found, the center can redirect the users' data to their corresponding BSs.
\\
In both the frameworks, since the transmitters do not have channel state information, an equal power allocation among $N_T$ data streams is assumed such that $\xx_l \sim \CN(0, \frac{P}{N_T}\II_{N_T})$, where $P$ is the transmit power constraint, assumed common, at all the BSs. Note that, if transmitter $l$ wants to convey $d$ data streams where $d < N_T$ it can employ an arbitrary precoding matrix $\VV_l \in \C^{N_T \times d}$ such that $\VV_l^H\VV_l = \II_d$ so that the system model remains statistically identical.

 
\section{Measure of Alignment}
\label{moa}

Recall that in a $K$-transmitter MIMO IC, each user will receive $(K-1)$ interference signals and one desired signal, each of dimension $N_T$. In order to quantify the suitability of a user for data transmission in the interference channel, we need to define a measure of alignment which quantifies the extent to which the received interference signals are aligned at each user. To that end, we will briefly review Grassmann manifold.

\subsection{Grassmann Manifold}
The Grassmann manifold $\Grass$ has been defined as the set of all $M$-dimensional subspaces of complex Euclidean $N$-dimensional space, $\C^{N}$ \cite{conway1996}. It is a widely used geometric concept in wireless communications and helps in the design and analysis of different methodologies. Let $\Aa \in \Grass$ be a $M$-dimensional subspaces. A $N \times M$ matrix $\aA$ is defined as \textit{generator matrix} for $\Aa$ if its columns span the subspace corresponding to $\Aa$ and forms an orthonormal bases for the same, i.e., $\aA^H\aA = \II_M$. Note that there can be infinitely many generator matrices for a given subspace $\Aa$ which can be obtained by the transformations $\aA \rightarrow \aA\UU$ where $\UU \in \C^{M \times M}$ is an arbitrary unitary matrix. For subspaces $\Aa, \Bb \in \Grass$, the $M$ ordered \textit{principal angles}, $\theta_1, \theta_2, \dots \theta_M \in [0, \pi/2]$ between the subspaces are obtained sequentially as %
\begin{align}
\label{eq:pangles}
&\text{cos} (\theta_m) = \underset{\begin{subarray}{c}
	\aaa \in \mathcal{A} \\
	\bbb \in \mathcal{B}
	\end{subarray}}{\text{max}} \: |\aaa^H \bbb| = |\aaa_m^H \bbb_m| \nonumber \\
& \hspace{55pt}  \text{s.t.} \quad ||\aaa || = ||\bbb || = 1  \nonumber \\
& \hspace{80pt}	\aaa^H \aaa_n = \bbb^H \bbb_n = 0 \: \forall \: n \in \mathcal{M}_m 
\end{align}

where $\mathcal{M}_m = \{1, 2, \dots, m-1\}$, while $\{\aaa_m\}_{m=1}^M$ and $\{\bbb_m\}_{m=1}^M$ are the \textit{principal vectors} for $\mathcal{A}$ and $\mathcal{B}$, respectively. 

The chordal distance between $\mathcal{A}$ and  $\mathcal{B}$ is defined as
\begin{equation}
d_c(\mathcal{A}, \mathcal{B}) = \sqrt{\sum_{m=1}^{M}\text{sin}^2(\theta_m)}
\end{equation}

Alternatively, chordal distance between the two subspaces $\mathcal{A}$ and $\mathcal{B}$ can be represented in terms of their generator matrices as
\begin{align}
\label{chordaldef}
d_c(\mathcal{A}, \mathcal{B}) = d_c(\aA, \bB) = \frac{1}{\sqrt{2}}||\aA\aA^H - \bB \bB^H||_F \\
= \sqrt{M -  \text{tr}(\aA^H\bB\bB^H\aA)}
\end{align}

 Chordal distance is known to be proportional to the degree of orthogonality between the subspaces. Note that, the chordal distance is invariant to the choice of generator matrices.

\subsection{Spread of Subspaces}
Let $\{\HH_l\}_{l=1}^L$ be $L$ matrices of size $N \times M$. Thus, each of these matrices would correspond to planes/subspaces $\{\mathcal{H}_l\}_{l=1}^L \in \Grass$ with $\{\hHH_l\}_{l=1}^L \in \C^{N \times M}$ as their generator matrices. In order to define the spread of these subspaces, consider the following 
\begin{align}
\label{eq:mean}
\FF = \underset{\FF' \text{ s.t. } \FF'^H\FF' = \II_M}{\operatorname{arg min}} \: \: \sum_{l = 1}^L d_c^2(\FF', \hHH_l)
\end{align}

Let $\mathcal{F} \in \Grass$ be the plane corresponding to $\FF$. Thus, $\mathcal{F}$ can be considered as the mean of subspaces $\{\mathcal{H}_i\}_{i=1}^L$. Quite naturally, we can define the spread of these subspaces as
\begin{align}
\label{eq:variance}
f &= \underset{\FF' \text{ s.t. } \FF'^H\FF' = \II_M}{\operatorname{\min}} \: \: \sum_{l = 1}^L d_c^2(\FF', \hHH_l)
\end{align}
The problem in (\ref{eq:mean}) can be simplified as follows :
\begin{align}
\FF &= \underset{\FF' \text{ s.t. } \FF'^H\FF' = \II_M}{\operatorname{arg min}} \: \: \sum_{l = 1}^L d_c^2(\FF', \hHH_l) \nonumber\\
&= \underset{\FF' \text{ s.t. } \FF'^H\FF' = \II_M}{\operatorname{arg min}} \: \: \sum_{l = 1}^L \left(M - \text{Tr}(\FF'^H\hHH_l\hHH_l^H\FF') \right)\nonumber\\
&= \underset{\FF' \text{ s.t. } \FF'^H\FF' = \II_M}{\operatorname{arg max}} \: \: \sum_{l = 1}^L \text{Tr}(\FF'^H\hHH_l\hHH_l^H\FF') \nonumber\\
&= \underset{\FF' \text{ s.t. } \FF'^H\FF' = \II_M}{\operatorname{arg max}} \: \:  \text{Tr}\left(\FF'^H\left(\sum_{l = 1}^L \hHH_{n,l}\hHH_{n,l}^H\right)\FF'\right) \nonumber\\
&= \left[\vv_1(\bbB_n) \quad \vv_2(\bbB_n) \: \ldots \: \vv_M(\bbB_n)\right]
\end{align}

where $\bbB = \sum\limits_{l=1}^{K}\hHH_l\hHH_l^H$ and $\vv_l(\HH)$ denotes the left singular vector of $\HH$ which corresponds to the $l$th largest singular value \cite{matrixhandbook}. Thus, $f$ in (\ref{eq:variance}) is given by 
\begin{align}
f &= \sum_{l = 1}^L d_c^2(\FF, \hHH_l) \nonumber \\
\Rightarrow f &= LM - \text{Tr}\left(\FF^H\left(\sum_{l = 1}^L \hHH_{n,l}\hHH_{n,l}^H\right)\FF\right) \nonumber \\
\Rightarrow f &= LM - \sum_{m=1}^{M}\lambda_m\left(\sum_{l = 1}^L\hHH_l\hHH_l^H\right) \nonumber\\
\Rightarrow f &\stackrel{(i)}{=} \sum_{m=M+1}^{2M}\lambda_m\left(\sum_{l = 1}^L\hHH_l\hHH_l^H\right)
\label{eq:moadef1}
\end{align}

where $\lambda_m(\HH)$ denotes the $m$th largest singular value of $\HH$ and $(i)$ holds because 
\begin{gather}
\sum_{m=1}^{2M}\lambda_m\left(\sum_{l = 1}^L\hHH_l\hHH_l^H\right) = \text{Tr}\left(\sum_{l = 1}^L\hHH_l\hHH_l^H\right) \nonumber\\
= \sum_{l = 1}^L\text{Tr}\left(\hHH_l\hHH_l^H\right) = LM
\end{gather}

 Smaller values of $f$ imply that the interfering subspaces are closely aligned and in the case of perfect alignment, $f=0$. In view of this, if the matrices $\{\HH_l\}_{l=1}^L$ correspond to the channels between a user and interfering base stations, $f$ can be defined as a measure of alignment for the user. The computation of the mean $\FF$ and measure of alignment $f$ involve singular value decomposition (SVD) and hence are expensive to compute. In what follows, we will explore approximations for the measure of alignment function $f$.
 
\subsubsection{$3$-Transmitter Case}
In the case of $3$-transmitter interference channels, each user has $2$ interference subspaces and $f$ has the form $\sum\limits_{m=M+1}^{2M}\lambda_m(\hHH_1\hHH_1^H + \hHH_2\hHH_2^H)$. In the following Lemma, we will find closed form expression for the eigenvalues of $\hHH_1\hHH_1^H + \hHH_2\hHH_2^H$.

\begin{Lem}
	\label{eigpangles}
	If $\hHH_1, \hHH_2 \in \C^{2M \times M}$ are the generator matrix of the subspaces $\mathcal{H}_1, \mathcal{H}_2 \in \Grassss$, eigenvalues of $\hHH_1\hHH_1^H + \hHH_2\hHH_2^H$ can be represented in descending order as
	\begin{align}
	\underbrace{1 + \cos(\theta_1), \ldots, 1 + \cos(\theta_M)}_{M}, \underbrace{1 - \cos(\theta_M), \ldots, 1 - \cos(\theta_1)}_{M}
	\end{align}
	where $\theta_m$ is the $m$th smallest principal angle between $\mathcal{H}_1$ and $\mathcal{H}_2$.
\end{Lem}

\begin{proof}
	We will first show that eigenvalues of $\sum\limits_{i=1}^{L} \hHH_i\hHH_i^H$ are invariant under a common rotation transformation on the corresponding subspaces $\{\mathcal{H}_i\}_{i=1}^L \in \Grass$, i.e., $\hHH \rightarrow \OO\hHH$ where $\OO \in \C^{2M \times 2M}$ is an arbitrary member of the orthogonal group $\mathcal{O}(2M)$ \cite{conway1996}.
	
	Let $\left(\sum\limits_{i=1}^{L} \hHH_i\hHH_i^H\right)\vv = \lambda\vv$ and consider
	 \begin{align*}
	 &\left(\sum\limits_{i=1}^{L} (\OO\hHH_i)(\OO\hHH_i)^H\right)\vv{'} = \lambda'\vv{'} \\
	\Rightarrow  & \OO\left(\sum\limits_{i=1}^{L} \hHH_i\hHH_i^H\right)\OO^H\vv' = \lambda'\vv{'}
	 \end{align*}
	 On right multiplying $\OO^H$ on both sides and using the fact that $\OO^H\OO = \II_N$, we have
	 \begin{align*}
	 &(\OO^H\OO)\left(\sum\limits_{i=1}^{L} \hHH_i\hHH_i^H\right)\OO^H\vv = \lambda'\OO^H\vv{'} \\
	 \Rightarrow  &\left(\sum\limits_{i=1}^{L} \hHH_i\hHH_i^H\right)\left(\OO^H\vv\right) = \lambda'(\OO^H\vv{'}) 
	 \end{align*}
	 
	 Comparing with $\left(\sum\limits_{i=1}^{L} \hHH_i\hHH_i^H\right)\vv = \lambda\vv$, we have that $\lambda = \lambda'$ and $\vv = \OO^H\vv'$.
	 
	 Thus, it follows that
	 	\begin{align*}
	 	 \lambda_m\left(\hHH_1\hHH_1^H + \hHH_2\hHH_2^H\right) = \lambda_m\big((\OO\hHH_1)(\OO\hHH_1)^H \\+ (\OO\hHH_2)(\OO\hHH_2)^H\big)
	 	 \end{align*}
	 	 
	 	 Let columns of $\hHH_1$ and $\hHH_2$ be the corresponding principal vectors, $\{\mathbf{a}_i\}_{i=1}^M$ and $\{\mathbf{b}_i\}_{i=1}^M$, respectively. Define ${\hHH '}_1 = \OO{\hHH}_1$ and ${\hHH '}_2 = \OO{\hHH '}_2$. We can choose $\OO \in \mathcal{O}(2M)$ such that ${\hHH '}_1$ and ${\hHH '}_2$ become \cite{conway1996}%
	 	 \begin{gather}
	 	 {\hHH '}_1 =
	 	 \begin{bmatrix}
	 	 \II_{M \times M} \\
	 	 \hline
	 	 \mathbf{0}_{M \times M}
	 	 \end{bmatrix} \\
	 	 {\hHH '}_2 =
	 	 \begin{bmatrix}
	 	 \text{cos}(\theta_1) & 0 & \cdots & 0 \\
	 	 0 & \text{cos}(\theta_2) & \cdots & 0 \\
	 	 \vdots  & \vdots  & \ddots & \vdots  \\
	 	 0 & 0 & \cdots & \text{cos}(\theta_M)\\
	 	 \hline
	 	 \text{sin}(\theta_1) & 0 & \cdots & 0 \\
	 	 0 & \text{sin}(\theta_2) & \cdots & 0 \\
	 	 \vdots  & \vdots  & \ddots & \vdots  \\
	 	 0 & 0 & \cdots & \text{sin}(\theta_M)
	 	 \end{bmatrix}
	 	 \end{gather}
	 
	 	Thus ${\hHH '}_1 {{\hHH}_1}^{'H} + {\hHH '}_2{\hHH}_2^{'H}$ has the  structure as in (\ref{eq:lemma1aux}). The eigenvalues of the matrix in (\ref{eq:lemma1aux}) can be trivially found to be : $1 + \cos(\theta_1)$, $1 + \cos(\theta_2)$, $\ldots$, $1 + \cos(\theta_M)$, $1 - \cos(\theta_M)$, $1 - \cos(\theta_{M-1})$, $\ldots$, $1 - \cos(\theta_1)$. This completes the proof.
\end{proof}
\begin{figure*}[t!]
	\normalsize
	\begin{align}
		\label{eq:lemma1aux}
		 	&\sum\limits_{l=1}^{2}{\hHH '}_l{\hHH_l}^{'H} =
		 	\begin{bmatrix}[ccc|ccc]
		 	1 + \text{cos}^2(\theta_1) & \cdots & 0 & \text{sin}(2\theta_1)/2 & \cdots & 0 \\
		 	\vdots  & \ddots & \vdots  & \vdots & \ddots & \vdots\\
		 	0 & 0 & 1 + \text{cos}^2(\theta_M)& 0 & \cdots & \text{sin}(2\theta_M)/2\\
		 	\hline
		 	\text{sin}(2\theta_1)/2 & \cdots & 0 & \text{sin}^2(\theta_1) & \cdots & 0\\
		 	\vdots  & \vdots  & \vdots & \vdots  & \ddots & \vdots\\
		 	0 & \cdots & \text{sin}(2\theta_M)/2)& 0 & \cdots & \text{sin}^2(\theta_M)
		 	\end{bmatrix}
	\end{align}
	\hrulefill
	\vspace*{4pt}
\end{figure*}

It follows from Lemma \ref{eigpangles} that
{
	\begin{align*}
	f = \sum\limits_{m=M+1}^{2M}\lambda_m(\hHH_1\hHH_1^H + \hHH_2\hHH_2^H) =\sum\limits_{m=1}^{M} 1 - \text{cos}(\theta_m)
	\end{align*}
}
Since $\text{cos}(\theta_m) \geq \text{cos}^2(\theta_m) \: \forall \: \theta_m  \in [0, \pi/2]$, we have that
{
	\begin{gather*}
	\sum\limits_{m=1}^{M} 1 - \text{cos}(\theta_m) \leq \sum\limits_{m=1}^{M} 1 - \text{cos}^2(\theta_m)\\
	\Rightarrow f \leq d_c^2(\hHH_1, \hHH_2) 
	\end{gather*}
}	 
We can redefine $f$ such that $f = d_c^2(\hHH_1, \hHH_2)$, which is intuitive as $d_c^2(\hHH_1, \hHH_2)$ is proportional to the degree of orthogonality between the corresponding subspaces and thus has the property that its value decreases as the interfering subspaces get closer or more aligned.

\subsubsection{General Case}
Unlike the $3$-transmitter case, it is difficult to obtain closed-form expressions for eigenvalues in the general case where the network has $K$ transmitters. In a $K$ transmitter network, there are $(K-1)$ interference subspaces at each user and thus $f$ has the form $\sum\limits_{m=M+1}^{2M}\lambda_m\left(\sum\limits_{l = 1}^{K-1}\hHH_l\hHH_l^H\right)$. In the following Lemma, we extend the bound for the general case of $K$-transmitter interference channels.
\begin{Lem} 
	For $K$-transmitter interference channels, the measure of alignment $f$ is bounded above by
	$\underset{1\leq j \leq K-1}{\operatorname{ min}}\sum\limits_{l=1}^{K-1} d_c^2(\hHH_j, \hHH_l)$
\end{Lem}

\begin{proof}
	
	Let $\FF$ be the mean of $\{\hHH_l\}_{l=1}^K$. Since $\FF$ is the minimizer of $\sum\limits_{l = 1}^{K-1} d_c^2(\FF', \HH_{l})$, it follows for any arbitrary $1 \leq j \leq K-1$ that
		\begin{gather}
		\sum\limits_{l=1}^{K-1} d_c^2(\FF, \hHH_{l})  \leq \sum\limits_{l=1}^{K-1} d_c^2(\hHH_{j}, \hHH_{l}) \nonumber\\
		\Rightarrow f \leq \sum\limits_{l=1}^{K-1} d_c^2(\hHH_{j}, \hHH_{l})
		\label{eq:boundproof}
		\end{gather}
		
	Since (\ref{eq:boundproof}) holds for any arbitrary $j$, it must hold for all $j \in \{1,\ldots, K-1\}$ and thus
		\begin{equation}
		f \leq \underset{1 \leq j \leq K-1}{\operatorname{ min}}\sum\limits_{l=1}^{K-1} d_c^2(\hHH_{j}, \hHH_{l})
		\end{equation}
	This completes the proof.
\end{proof}

From (\ref{eq:moadef1}), we have $f = \sum\limits_{l=1}^{K-1} d_c^2(\FF, \hHH_l)$, where $\FF$ corresponds to the mean of the subspaces corresponding to $\{\hHH_l\}_{l=1}^{K-1}$. Since computing the mean $\FF$ or even $f$ (directly) is computationally prohibitive, we can approximate the mean $\FF$ by an element in $\{\hHH_l\}_{l=1}^{K-1}$ which is nearest to it. Indeed $\hHH_{\hat{j}}$ where $\hat{j} = \underset{1\leq j \leq K-1}{\operatorname{\min}}{\operatorname{arg min}}\sum\limits_{l=1}^{K-1} d_c^2(\hHH_j, \hHH_l)$ is closest to the mean $\FF$ and we can redefine the measure of alignment as follows 
\begin{align}
f = \underset{1\leq j \leq K-1}{\operatorname{\min}}\sum\limits_{l=1}^{K-1} d_c^2(\hHH_j, \hHH_l)
\end{align}

Note that this approximation to the actual measure of alignment is cheaper to compute. Also, for $K=3$, the above expression reduces to the one obtained for the $3$-transmitter interference channel.


\section{Opportunistic Interference Alignment through User Selection}
\label{oiaus}
In this section, we consider the user selection problem (refer Section \ref{systemmodelus}) in which one user is selected in each cell such that opportunistic IA is achieved. The $n$th user in cell $k$ calculates the measure of alignment function $f_n^k$ as follows
\begin{align}
f_n^k = \underset{\stackrel{1\leq j \leq K}{j \neq k}}{\operatorname{\min}}\sum\limits_{l=1, l \neq k}^{K} d_c^2(\hHH_{n,j}^k, \hHH_{n,l}^k)
\end{align}

where $\hHH_{n,l}^k$ is an arbitrary generator matrix for $\HH_{n,l}^k$. Following this, each user feeds the measure of alignment back to its corresponding transmitter. After receiving this information from their users, transmitter $k$ selects the user, $n_k^*$, with the minimum value of measure of alignment
\begin{align}
n_k^* = \underset{1\leq n \leq K}{\operatorname{arg \min}} \: f_n^k
\end{align}

The selected user $n_k^*$ in cell $k$ employs the post-processing matrix $\UU_{n_k^*}^k$ which minimizes the interference leakage \cite{gomadam2011} as follows
\begin{align}
\label{eq:postus}
\UU_{n_k^*}^k &=\underset{\UU}{\operatorname{arg min}} \: \: \text{Tr}\left(\UU^H\left(\sum\limits_{l=1, l \neq k}^{K}\HH_{n_k^*,l}^k(\HH_{n_k^*,l}^k)^H\right)\UU\right) \nonumber \\
&= \left[\vv_{M+1}\left(\bB_{n_k^*}\right), \: \vv_{M+2}\left(\bB_{n_k^*}\right), \: \ldots, \: \vv_{2M}\left(\bB_{n_k^*}^k\right)\right]
\end{align}
where $\bB_{n_k^*}^k = \sum\limits_{l=1, l\neq k}^{K}\HH_{{n_k^*},l}^k(\HH_{{n_k^*},l}^k)^H$. The achievable sum-rate \cite{OIA2} for the network is given by 
\begin{align}
\label{oiaussumrate}
R_{sum} =\sum\limits_{k=1}^{K}\log_2\frac{\Big|\II_{M} + \frac{P}{M}\sum\limits_{l=1}^{K}\UU_{n_k^*}^H\HH_{n_k^*,l}^k(\HH_{n_k^*,l}^k)^H\UU_{n_k^*}\Big|}{\Big|\II_{M} + \frac{P}{M}\sum\limits_{{l=1},\:{l \neq k}}^{K}\UU_{n_k^*}^H\HH_{n_k^*,l}^k(\HH_{n_k^*,l}^k)^H\UU_{n_k^*}\Big|}
\end{align}

\section{Opportunistic Interference Alignment through User Pairing}
\label{oiaup}

In this section, we consider the problem of finding transmitter-receiver pairing configuration (refer Section \ref{systemmodelup}) in order to achieve IA opportunistically. For OIA in the user pairing framework, the receivers feed back the measure of alignment to a central node, which in turn decides the pairing configuration.

Each user receives $K$, $M$-dimensional signals among which atmost one can be the desired signal. Unlike the OIA with user selection case, the desired signal is not predefined and it will depend on the channel conditions for all the users in the network. The measure of alignment at user $n$ when it is paired with the $k$th BS can be defined as
\begin{align}
\label{moaeq}
f_{n,k} = \underset{\stackrel{1\leq j \leq K}{j \neq k}}{\operatorname{\min}}\sum\limits_{l=1, l \neq k}^{K} d_c^2(\gGG_{n,j}, \gGG_{n,l})
\end{align}

where $\gGG_{n,l}$ is an arbitrary generator matrix for $\GG_{n,l}$. Each user thus computes $K$ measure of alignment functions, $\{f_{n,k}\}_{k=1}^K$ corresponding to each BS.

Let us define the vector of measure of alignment at user $n$, $\ff_n$ as
\begin{equation}
\ff_n = [f_{n,1}, \: f_{n,2}, \: \ldots, \: f_{n,K}]^T
\end{equation}
Each user feedbacks its corresponding measure of alignment vectors $\{\ff_n\}_{n=1}^N$ to a central node. The center aggregates the data from all the users and forms the $N \times K$ feedback matrix, $\FF$ defined as
\begin{equation}
\FF = [\ff_1, \: \ff_2, \: \ldots, \: \ff_N]^T
\end{equation}

Each entry in the matrix $\FF$ corresponds to a pair in the original network. The smaller the value of the entry, the more likely it is for the corresponding link to have the interferences aligned and thus more likely to be chosen in the final user pairing solution. Having obtained the matrix $\FF$, the center can choose $K$ non-conflicting pairs which constitute the minimum sum for the measure of alignment. Therefore, the optimization problem can be formulated as
\begin{subequations}
\label{up_opti}
\begin{align}
\underset{\mathbf{P}}{\text{min}} & \quad \sum_{i=1}^{N}\sum_{j=1}^{K}P_{i,j}f_{i,j} \\
\text{subject to} & \quad \sum_{j}P_{i,j} \leq 1 \; \forall \: i\\
& \quad \sum_{i}P_{i,j} = 1 \; \forall \: j  \\
& \quad P_{i,j} \in \{0, 1\} \; \forall \: i, j 
\end{align}
\end{subequations}
This optimization can be solved efficiently by the rectangular Hungarian algorithm \cite{bourgeois1971}. After the optimal pairing configuration $\mathbf{P}^*$ has been found, each user which is connected to a BS can employ a post-processing matrix which minimizes the interference leakage similar to (\ref{eq:postus}). Let $n_k^*$ be the user paired with BS $k$, i.e., $P^*_{n_k^*,k} = 1$. The expression for achievable sum-rate will be same as (\ref{oiaussumrate}).

\begin{figure}[t]
	\label{fig:K3sumrate}
	\includegraphics[width=0.5\textwidth]{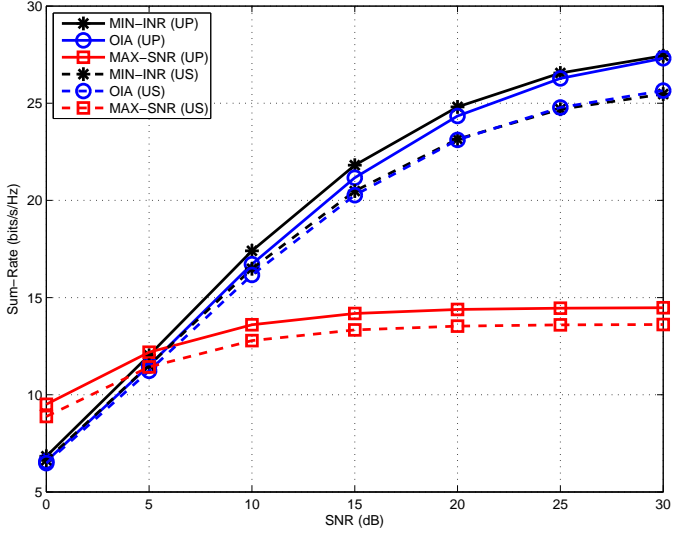}
	\caption{$3$-transmitter MIMO IC with $M = 3$ and $N=30$ }
\end{figure}

\begin{figure}[t]
	\label{fig:K3complexity}
	\includegraphics[width=0.5\textwidth]{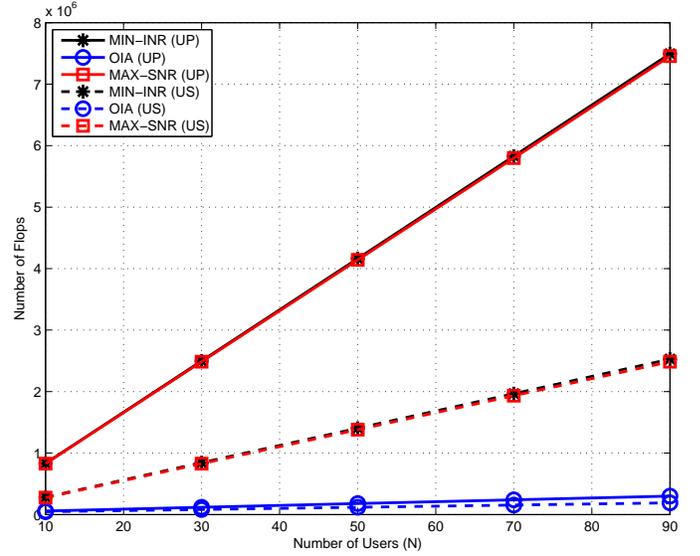}
	\caption{Complexity for $3$-transmitter MIMO IC with $M = 3$ }
\end{figure}

\section{Performance Comparison}
\label{results}
In this section, we compare the performance in terms of sum-rate and computational complexity of the proposed OIA algorithm with the conventional MAX-SNR and MIN-INR schemes \cite{OIA1, OIA2}. MAX-SNR and MIN-INR have been proposed for the user selection framework in \cite{OIA1, OIA2}. We extend them for the user pairing framework as done for OIA in Section \ref{oiaup}. In what follows, US and UP denote user selection and user pairing, respectively.

\subsection{Complexity Analysis}
\label{complexity}
In this section, we will discuss the computational complexity of the algorithms using flop counts. The complexity of an operation is counted as total number of flops required which is defined as a real floating point operation and we denote it by $\psi$. The flop counts for some typical operations for a complex matrix $\aA \in \C^{m \times n}$ with $m \geq n$ are
\begin{subequations}
	\begin{gather}
	\label{flopoperations}
	\psi(\aA + \aA) =  2mn \\
	\psi(||\aA||_F) = 4mn \\
	\psi(\text{GSO}(\aA)) = 8mn^2 - 2mn \\
	\psi(\text{SVD}(\aA)) = 24mn^2 + 48mn^2 + 54n^3 \\
	\psi(\text{MUL}(\aA)) = \psi(\aA\aA^H) = 8mn^2 - 2mn
	\end{gather}
\end{subequations}
where GSO stands for Gram-Schmidt orthogonalization, SVD stands for singular value decomposition and MUL$(\aA)$ denotes the operation $\aA\aA^H$.

The proposed OIA in the user selection framework requires $(K-1)$ GSO operations, $(K-1)$ MUL operations and ${{K-1} \choose 2}$ matrix subtractions as well as $||.||_F$ operations at each user in every cell. At the selected user, MUL operation for $(K-1)$ times, $(K-2)$ matrix additions and a single SVD is required. Thus the total complexity for OIA in the user selection framework is given by 
\begin{align}
\psi_\text{OIA-US} =  K  \Big(&{S}\left(N_R^3(4K-4) + N_R^2(3K^2 - 11K + 8)\right) \nonumber \\
& + \left(N_R^3(124 + 2K) + N_R^2(K-3)\right) \Big)
\end{align}

OIA in the user pairing framework requires $K$ GSO operations, $K$ MUL operations and ${K \choose 2}$ matrix subtractions as well as $||.||_F$ operations at each user. At the $K$ selected users, MUL operation for $(K-1)$ times, $(K-2)$ matrix additions and a single SVD is required. Thus the total complexity for OIA in the user pairing framework is given by 
\begin{align}
\psi_\text{OIA-UP} = \:  N \times \left(N_R^3(4K) + N_R^2(3K^2-5K)\right) + \nonumber \\
K \times \left(N_R^3(124 + 2K) + N_R^2(K-3)\right)
\end{align}

The complexity for MAX-SNR and MIN-INR in the user selection framework denoted by $\psi_{\text{MAX-SNR-US}}$ and $\psi_{\text{MIN-INR-US}}$, respectively, is given in \cite{OIA2}. MIN-INR in the user pairing framework requires $K$ MUL operations, $2K$ matrix additions and a single SVD at every user. The total complexity for MIN-INR is given by
\begin{align}
\psi_\text{MIN-INR-UP} = \:  N \times \left(N_R^3(128K) + N_R^2(3K)\right) 
\end{align}

The MAX-SNR scheme requires $K$ GSO operations and $K$ SVD at every user in the user pairing framework. Thus, the total complexity for MAX-SNR is given by
\begin{align}
\psi_\text{MAX-SNR-UP} = \:  N \times \left(N_R^3(128K) - N_R^2(K)\right) 
\end{align}

Note that we have ignored the complexity of solving the optimization problem in \eqref{up_opti} which arises in the user pairing framework. This is because the computation happens only once at the center and not at the mobile users.

Fig.~$4$ and Fig.~$6$ show the plot of computational complexity vs. total number of users $N$ for a $3$-transmitter MIMO IC with $M=3$ and a $4$-transmitter MIMO IC with $M=6$, respectively. It can be observed that the complexity of OIA is only a small fraction of the complexity of MIN-INR and MAX-SNR schemes. Moreover, user pairing when compared to user selection has roughly the same complexity in case of proposed OIA, but the same is not true for both, MIN-INR and MAX-SNR.

\subsection{Sum-Rate}
Fig.~$3$ and Fig.~$5$ show the sum-rate vs. signal-to-noise ratio (SNR) plot for the proposed OIA and the conventional schemes for a $3$-transmitter MIMO IC with $M=3$ and a $4$-transmitter MIMO IC with $M=6$, respectively. It can be observed that the proposed OIA achieves sum-rates close to MIN-INR but at a significantly lower computational complexity. Moreover, with the same number of total users in the network, user pairing outperforms user selection.

\begin{figure}[t!]
	\label{fig:K4sumrate}
	\includegraphics[width=0.5\textwidth]{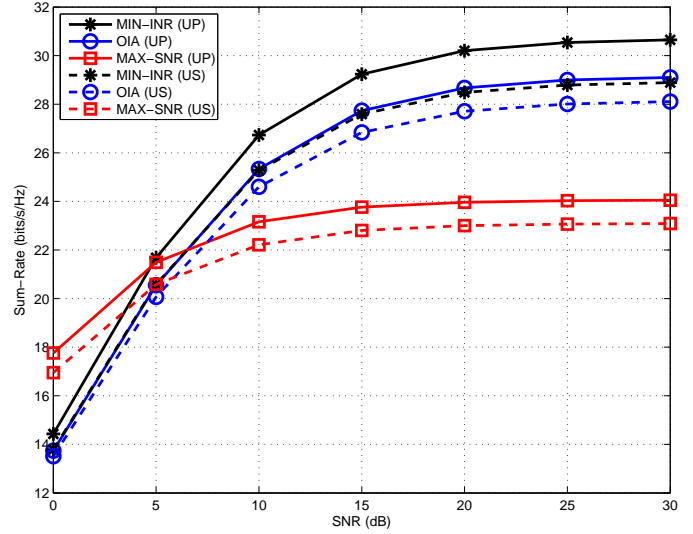}
	\caption{$4$-transmitter MIMO IC with $M = 6$ and $N=40$ }
\end{figure}

\begin{figure}[th!]
	\label{fig:K4complexity}
	\includegraphics[width=0.5\textwidth]{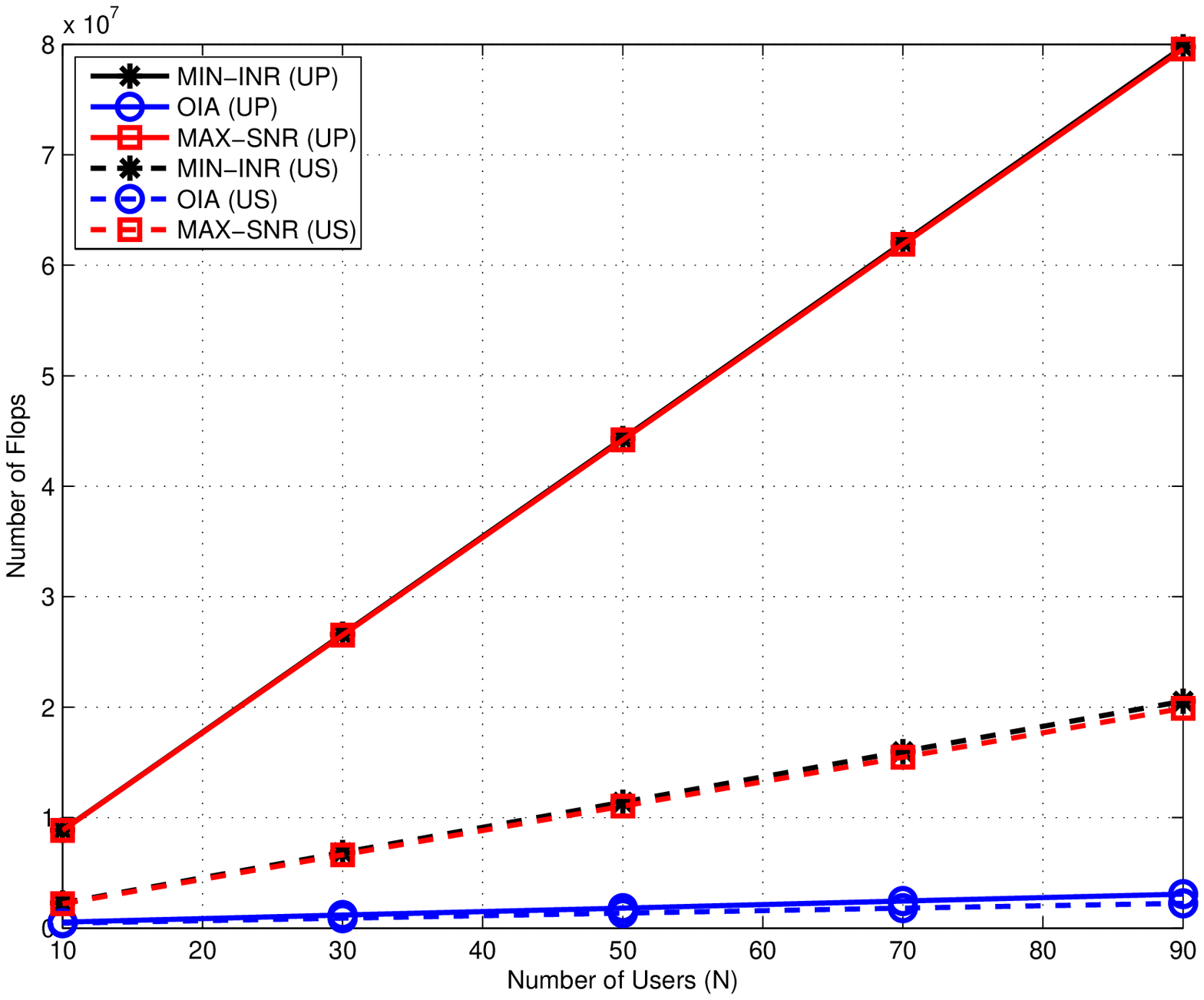}
	\caption{Complexity for $4$-transmitter MIMO IC with $M = 6$}
\end{figure}

\begin{figure}[t!]
	\label{fig:ratevsN}
	\includegraphics[width=0.51\textwidth]{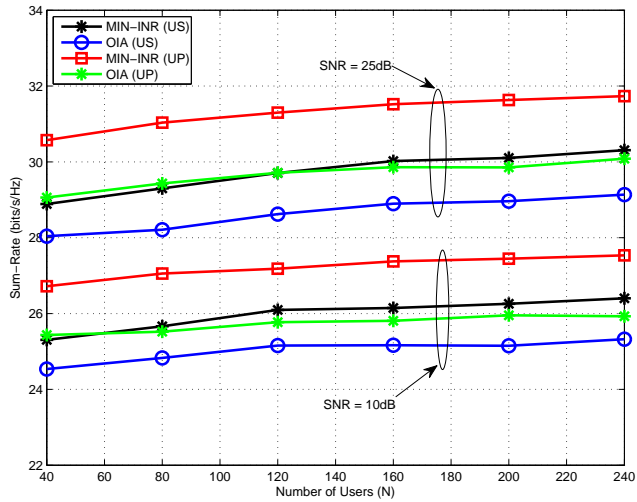}
	\caption{$4$-transmitter MIMO IC with $M = 6$ and fixed SNR}
\end{figure}

\begin{figure}[t!]
	\label{fig:ratevsT}
	\includegraphics[width=0.51\textwidth]{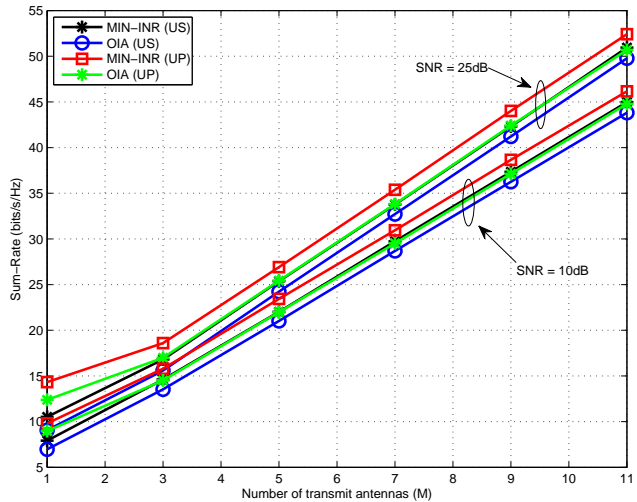}
	\caption{$4$-transmitter MIMO IC with $N = 100$ and fixed SNR}
\end{figure}

Fig.~$7$ shows the plot of sum-rate vs. the total number of users for a $4$-transmitter MIMO IC with $M=6$ and SNRs of $10$ dB and $25$ dB. As expected, the performance of all the algorithms improve as the number of users is increased. Also, user pairing provides more than a $4$-fold gain over user selection in terms of total number of users required to achieve similar sum-rate performance. With respect to the number of users, the gaps in sum-rate performance for different algorithms is almost constant. Fig.~$8$ shows the plot of sum-rate vs. the number of antennas $M$ for a $4$-transmitter MIMO IC with $N=100$ and SNRs of $10$ dB and $25$ dB. It can be observed that the sum-rate increases almost linearly with the number of transmit antennas for all the algorithms.
\section{Conclusion}
\label{conclusion}
In this paper, we have considered two different system models, namely, user selection and user pairing for $K$-transmitter MIMO interference channels. By exploiting multiuser diversity, we propose low complexity opportunistic interference alignment (OIA) algorithms for both the models. The proposed OIA algorithms are compared with conventional schemes, MIN-INR and MAX-SNR, and found to achieve comparable sum-rates but at a significantly reduced computational complexity.

\bibliographystyle{IEEEtran}
\bibliography{IEEEabrv,paper}

\end{document}